\documentclass[conference,a4paper,twoside]{IEEEtran}

\usepackage[cmex10]{amsmath}
\usepackage{graphicx,epic,eepic,epsfig,latexsym,amssymb,verbatim,color,revsymb}

\usepackage{theorem}
\newtheorem{definition}{Definition}
\newtheorem{proposition}[definition]{Proposition}
\newtheorem{lemma}[definition]{Lemma}

\newtheorem{corollary}[definition]{Corollary}

\newtheorem{observation}[definition]{Observation}
\newtheorem{question}[definition]{Question}

\def\squareforqed{\hbox{\rlap{$\sqcap$}$\sqcup$}}
\def\qed{\ifmmode\squareforqed\else{\unskip\nobreak\hfil
\penalty50\hskip1em\null\nobreak\hfil\squareforqed
\parfillskip=0pt\finalhyphendemerits=0\endgraf}\fi}
\def\endenv{\ifmmode\;\else{\unskip\nobreak\hfil
\penalty50\hskip1em\null\nobreak\hfil\;
\parfillskip=0pt\finalhyphendemerits=0\endgraf}\fi}
\newenvironment{proof}{\noindent \textbf{{Proof~} }}{\qed}
\newenvironment{remark}{\noindent \textbf{{Remark~}}}{}

\mathchardef\ordinarycolon\mathcode`\:
\mathcode`\:=\string"8000
\def\vcentcolon{\mathrel{\mathop\ordinarycolon}}
\begingroup \catcode`\:=\active
  \lowercase{\endgroup
  \let :\vcentcolon
  }

\newcommand{\nc}{\newcommand}
\nc{\rnc}{\renewcommand}
\nc{\beq}{\begin{equation}}
\nc{\eeq}{{\end{equation}}}
\nc{\beqa}{\begin{eqnarray}}
\nc{\eeqa}{\end{eqnarray}}
\nc{\lbar}[1]{\overline{#1}}
\nc{\bra}[1]{\langle#1|}
\nc{\ket}[1]{|#1\rangle}
\nc{\ketbra}[2]{|#1\rangle\!\langle#2|}
\nc{\braket}[2]{\langle#1|#2\rangle}
\nc{\proj}[1]{| #1\rangle\!\langle #1 |}
\nc{\avg}[1]{\langle#1\rangle}
\nc{\Rank}{\operatorname{Rank}}
\nc{\smfrac}[2]{\mbox{$\frac{#1}{#2}$}}
\nc{\tr}{\operatorname{Tr}}
\nc{\ox}{\otimes}
\nc{\dg}{\dagger}
\nc{\dn}{\downarrow}
\nc{\cA}{{\cal A}}
\nc{\cB}{{\cal B}}
\nc{\cC}{{\cal C}}
\nc{\cD}{{\cal D}}
\nc{\cE}{{\cal E}}
\nc{\cF}{{\cal F}}
\nc{\cG}{{\cal G}}
\nc{\cH}{{\cal H}}
\nc{\cI}{{\cal I}}
\nc{\cJ}{{\cal J}}
\nc{\cK}{{\cal K}}
\nc{\cL}{{\cal L}}
\nc{\cM}{{\cal M}}
\nc{\cN}{{\cal N}}
\nc{\cO}{{\cal O}}
\nc{\cP}{{\cal P}}
\nc{\cQ}{{\cal Q}}
\nc{\cR}{{\cal R}}
\nc{\cS}{{\cal S}}
\nc{\cT}{{\cal T}}
\nc{\cU}{{\cal U}}
\nc{\cX}{{\cal X}}
\nc{\cY}{{\cal Y}}
\nc{\cZ}{{\cal Z}}
\nc{\csupp}{{\operatorname{csupp}}}
\nc{\qsupp}{{\operatorname{qsupp}}}
\nc{\var}{{\operatorname{var}}}
\nc{\rar}{\rightarrow}
\nc{\lrar}{\longrightarrow}
\nc{\polylog}{{\operatorname{polylog}}}
\nc{\1}{{\openone}}
\nc{\wt}{{\operatorname{wt}}}
\nc{\av}[1]{{\left\langle {#1} \right\rangle}}

\nc{\RR}{{{\mathbb R}}}
\nc{\CC}{{{\mathbb C}}}
\nc{\FF}{{{\mathbb F}}}
\nc{\NN}{{{\mathbb N}}}
\nc{\ZZ}{{{\mathbb Z}}}
\nc{\PP}{{{\mathbb P}}}
\nc{\QQ}{{{\mathbb Q}}}
\nc{\UU}{{{\mathbb U}}}
\nc{\EE}{{{\mathbb E}}}
\nc{\id}{{\operatorname{id}}}

\nc{\CHSH}{{\operatorname{CHSH}}}

\nc{\be}{\begin{equation}}
\nc{\ee}{{\end{equation}}}
\nc{\bea}{\begin{eqnarray}}
\nc{\eea}{\end{eqnarray}}
\nc{\<}{\langle}
\rnc{\>}{\rangle}
\nc{\Hom}[2]{\mbox{Hom}(\CC^{#1},\CC^{#2})}
\nc{\rU}{\mbox{U}}

\nc{\ob}[1]{#1}

\nc{\SEP}{{\text{SEP}}}
\nc{\NS}{{\text{NS}}}
\nc{\LOCC}{{\text{LOCC}}}
\nc{\PPT}{{\text{PPT}}}
\nc{\EXT}{{\text{EXT}}}
\nc{\Sym}{{\operatorname{Sym}}}

\nc{\ERLO}{{E_{\text{r,LO}}}}
\nc{\ERLOCC}{{E_{\text{r,LOCC}}}}
\nc{\ERPPT}{{E_{\text{r,PPT}}}}
\nc{\ERLOCCinfty}{{E^{\infty}_{\text{r,LOCC}}}}
\nc{\Aram}{{\operatorname{\sf A}}}

\begin{document}

\title{Fully Quantum Arbitrarily Varying Channels: \protect\\ Random Coding Capacity and Capacity Dichotomy}

\date{9 January 2018}

\author{%
  \IEEEauthorblockN{Holger Boche and Christian Deppe}
  \IEEEauthorblockA{Theoretische Informationstechnik\\
                    and Nachrichtentechnik\\
                    TU M\"unchen, Germany\\
%                    {\footnotesize ${}^*$On leave from Universit\"at Bielefeld}\\[-1mm]
%                    {\footnotesize Fakult\"at f\"ur Mathematik}\\
                    \{boche,\,christian.deppe\}@tum.de}
  \and 
  \IEEEauthorblockN{Janis N\"otzel}
  \IEEEauthorblockA{Theoretische Nachrichtentechnik\\
                    TU Dresden, Germany\\
                    janis.noetzel@tu-dresden.de}
  \and 
  \IEEEauthorblockN{Andreas Winter}
  \IEEEauthorblockA{ICREA \&{} Departament de F\'isica\\
                    Grup d'Informaci\'o Qu\`antica\\
                    Universitat Aut\`onoma de Barcelona, Spain\\
                    andreas.winter@uab.cat}
}

%\author{%
%  \IEEEauthorblockN{Holger Boche\IEEEauthorrefmark{1},
%                    Christian Deppe\IEEEauthorrefmark{1}\IEEEauthorrefmark{2},
%                    Janis N\"otzel\IEEEauthorrefmark{3}
%                    and Andreas Winter\IEEEauthorrefmark{4}}
%  \IEEEauthorblockA{\IEEEauthorrefmark{1}%
%                    TUM, Germany\\
%                    Email: boche@tum.de}
%  \IEEEauthorblockA{\IEEEauthorrefmark{2}%
%                    Uni Bielefeld\\
%                    Email: cdeppe@uni-bielefeld.de}
%  \IEEEauthorblockA{\IEEEauthorrefmark{3}%
%                    TU Dresden\\
%                    Email: janis.noetzel@tu-dresden.de}
%  \IEEEauthorblockA{\IEEEauthorrefmark{4}%
%                    UAB \&{} ICREA\\
%                    Email: andreas.winter@uab.cat}
%}

\maketitle

\begin{abstract}
We consider a model of communication via a fully quantum jammer channel 
with quantum jammer, quantum sender and quantum receiver, which we dub
\emph{quantum arbitrarily varying channel (QAVC)}. Restricting
to finite dimensional user and jammer systems, we show, using
permutation symmetry and a de Finetti reduction, how the random
coding capacity (classical and quantum) of the QAVC is reduced
to the capacity of a naturally associated compound channel, which is 
obtained by restricting the jammer to i.i.d.~input states. 
Furthermore, we demonstrate that the shared randomness required is at
most logarithmic in the block length, using a random matrix tail bound.
This implies a dichotomy theorem: either the classical capacity of the
QAVC is zero, and then also the quantum capacity is zero, or 
each capacity equals its random coding variant.
\end{abstract}

\section{Fully quantum AVC and random codes}

We consider a simple, fully quantum model of arbitrarily varying channel
(QAVC). Namely, we have three agents, Alice (sender), Bob (receiver) and 
Jamie (jammer), each controlling a quantum system $A$, $B$ and $J$, respectively. 
The channel is simply a completely positive and trace preserving (cptp) map
$\cN : \cL(A\ox J) \longrightarrow \cL(B)$, and we assume it to
be memoryless on blocks of length $\ell$, 
i.e.~$\cN^{\ox \ell} : \cL(A^\ell\ox J^\ell) \longrightarrow \cL(B^\ell)$,
with $A^\ell = A \ox \cdots \ox A$ ($\ell$ times), etc. However, crucially,
neither Alice's nor Jamie's input states need to be tensor product
or even separable states.
We shall assume throughout that all three Hilbert spaces $A$, $B$
and $J$ have finite dimension, $|A|$, $|B|$, $|J| < \infty$.
The previously introduced AVQC model of Ahlswede and Blinovsky 
\cite{AhlswedeBlinovsky:AVQC-C}, and more generally Ahlswede 
\emph{et al.} \cite{ABBN:AVQC-Q}, is obtained by channels 
$\cN$ that first dephase the input $J$ in a fixed basis, so that the
choices of the jammer are effectively reduced to basis states
$\proj{j}$ of $J$ and their convex combinations.
Note that this generalises the classical AVC, which is simply
a channel with input alphabet $\cX \times \cS$ and output
alphabet $\cY$, given by transition probabilities $N(y|x,s)$,
and such a channel can always be interpreted as a cptp map.
This model has been considered in \cite{e-ass-Q,e-ass-C}, however
in those works principally from the point of view that Jamie is
helping Alice and Bob, passively, by providing a suitable input
state to $J$. Contrary to the classical AVC and the AVQC considered
in \cite{AhlswedeBlinovsky:AVQC-C,ABBN:AVQC-Q}, where the jammer
effectively always selects a tensor product channel between Alice 
and Bob, the fact that we allow general quantum inputs on $J^\ell$, 
including entangled states, permits Jamie to induce non-classical
correlations between the different channel systems. These correlations, 
as was observed in \cite{e-ass-Q,e-ass-C}, are not only highly nontrivial,
but can also have a profound impact on the communication capacity
of the channel between Alice and Bob.
In the present context, however, Jamie is fundamentally an adversary.

Define a \emph{(deterministic) classical code} for $\cN$ of block length $\ell$ 
as a collection $\cC = \{(\rho_m,D_m):m=1,\ldots,M\}$ of states
$\rho_m \in \cS(A^\ell)$ and POVM elements $D_m \geq 0$ acting on
$B^n$, such that $\sum_{m=1}^M D_m = \1$. Its rate is defined as
$\frac{1}{\ell}\log M$, the number of bits encoded per channel use.
Its error probability is defined as the average over uniformly 
distributed messages and with respect to a state $\sigma$ on $J^\ell$: 
\[
  P_{\text{err}}(\cC,\sigma) 
       := \frac1M \sum_{m=1}^M \tr \left(\cN^{\ox\ell}(\rho_m\ox\sigma)\right)\!(\1-D_m).
\]

For the transmission of quantum information, define a
\emph{(deterministic) quantum code} for $\cN$ of block length $\ell$ 
as a pair $\cQ = (\cE,\cD)$ of cptp maps
$\cE:\cL(\CC^L) \longrightarrow \cL(A^\ell)$ and
$\cD:\cL(B^\ell) \longrightarrow \cL(\CC^L)$. Its rate is
$\frac{1}{\ell}\log L$, the number of qubits encoded per channel use,
and the error is quantified, with respect to a state $\sigma$ on $J^\ell$, 
as the ``infidelity''
\[
  \widehat{F}(\cQ,\sigma) 
        := 1- \tr \bigl( (\id\ox\cD\circ\cN^{\ox\ell}_{\sigma}\circ\cE)\Phi_L \bigr)\!\cdot\!\Phi_L,
\]
with the maximally entangled state $\Phi_L = \frac1L \sum_{ij} \ketbra{ii}{jj}$.
Here, we have introduced the channels
$\cN_\sigma:\cL(A)\longrightarrow\cL(B)$ defined by fixing the
jammer's state to $\sigma$, $\cN_\sigma(\rho) := \cN(\rho\ox\sigma)$.

Note that we use the language of ``deterministic'' code, although
in quantum information this is indistinguishable from stochastic
encoders; it is meant to differentiate from ``random'' codes, which use
shared correlation:
A \emph{random classical [quantum] code} for $\cN$ of block length $\ell$ consists of a
random variable $\lambda$ with a well-defined distribution and 
a family of deterministic codes $\cC_\lambda$ [$\cQ_\lambda$]. 
%Note that a deterministic code is a special case of random code, 
%where $\lambda$ takes only one value with probability $1$.
%
The error probability if $(\cC_\lambda)$, always with respect to a 
state $\sigma$ on $J^\ell$, is simply the expectation over $\lambda$, 
i.e.~$\EE_\lambda P_{\text{err}}(\cC_\lambda,\sigma)$. The error of the 
random quantum code is similarly $\EE_\lambda \widehat{F}(\cQ_\lambda,\sigma)$.

The operational interpretation of the random code model 
is that Alice and Bob share knowledge of the random variable
$\lambda$, and use $\cC_\lambda$ accordingly, but that Jamie is ignorant
of it. This shared randomness is thus a valuable resource, which for
random codes is considered freely available, whose amount, however,
we would like to control at the same time.

The capacities associated to these code concepts are defined as usual, 
as the maximum achievable rate as block length goes to infinity and the 
error goes to zero:
\begin{align*}
  C_{\text{det}}(\cN)  &:= \limsup_{\ell\rightarrow\infty} \frac{1}{\ell}\log M
                               \text{ s.t. }
                               \sup_{\sigma} P_{\text{err}}(\cC,\sigma) \rightarrow 0, \\
  C_{\text{rand}}(\cN) &:= \limsup_{\ell\rightarrow\infty} \frac{1}{\ell}\log M
                               \text{ s.t. }
                               \sup_{\sigma} \EE_\lambda P_{\text{err}}(\cC_\lambda,\sigma) \rightarrow 0, \\
  Q_{\text{det}}(\cN)  &:= \limsup_{\ell\rightarrow\infty} \frac{1}{\ell}\log L
                               \text{ s.t. }
                               \sup_{\sigma} \widehat{F}(\cQ,\sigma) \rightarrow 0, \\
  Q_{\text{rand}}(\cN) &:= \limsup_{\ell\rightarrow\infty} \frac{1}{\ell}\log L
                               \text{ s.t. }
                               \sup_{\sigma} \EE_\lambda \widehat{F}(\cQ,\sigma) \rightarrow 0.
\end{align*}

If in the above error maximisations Jamie is restricted to tensor
power states $\sigma^{\ox\ell}$, the QAVC model becomes a compound channel: 
$\cN^{\ox\ell}_{\sigma^{\ox\ell}} = (\cN_\sigma)^{\ox\ell}$, $\sigma\in\cS(J)$.
Its classical and quantum capacities are denoted $C(\{\cN_\sigma\}_\sigma)$
and $Q(\{\cN_\sigma\}_\sigma)$, respectively.

\section{Random coding capacities: \protect\\ from QAVC to its compound channel}
By definition, (see also \cite{BB:compound,BBN:q-compound} and \cite{ABBN:AVQC-Q})
\begin{equation}\begin{split}
  \label{eq:cap-chains}
  C_{\text{det}}(\cN) &\leq C_{\text{rand}}(\cN) \leq C(\{\cN_\sigma\}_\sigma), \text{ and}\\
  Q_{\text{det}}(\cN) &\leq Q_{\text{rand}}(\cN) \leq Q(\{\cN_\sigma\}_\sigma).
\end{split}\end{equation}
Here, we show that for the random capacity, the rightmost inequalities
are identities, by proving bounds in the opposite direction.
For the quantum capacity, this was done in \cite[Appendix~A]{e-ass-Q}.
%which in turn is based on symmetrisation by random permutations from $S_\ell$
%and the so-called de Finetti reduction (``Postselection Lemma'')
%from~\cite{CKR}; see also the model considered in~\cite{DSW:feedback-0}.
To present the argument, define the permutation operator $U^\pi$ acting on the 
tensor power $A^\ell$ as permuting the subsystems, for a permutation $\pi\in S_\ell$:
\[
  U^\pi\bigl(\ket{\alpha_1}\ket{\alpha_2}\cdots\ket{\alpha_\ell}\bigr)
       = \ket{\alpha_{\pi^{-1}(1)}}\ket{\alpha_{\pi^{-1}(2)}}\cdots\ket{\alpha_{\pi^{-1}(\ell)}},
\]
which extends uniquely by linearity. This is a unitary representation of
the symmetric group, which is defined for any Hilbert space. 
The quantum channel obtained by the conjugation action of 
$U^\pi$ is denoted $\cU_\pi(\alpha) = U^\pi \alpha U^{\pi\dagger}$.

\begin{proposition}
  \label{prop:compound-to-AVC:q}
  Let $\cQ=(\cE,\cD)$ be a quantum code for the compound 
  channel $\{\cN_\sigma\}_{\sigma\in\cS(J)}$ at block length $\ell$ 
  of size $L$ and with fidelity $1-\epsilon$, i.e.~for all $\sigma\in\cS(J)$,
  \[
    \widehat{F}\left(\cQ,\sigma^{\ox\ell}\right)
        =    1-\tr \bigl( (\id\ox\cD\circ\cN^{\ox\ell}_{\sigma}\circ\cE)\Phi_L \bigr)\!\cdot\!\Phi_L
        \leq \epsilon.
  \]
  Then, the random quantum code $(\cQ_\pi)_{\pi\in S_\ell}$ with a uniformly
  distributed random permutation $\pi$ of $[\ell]$, defined by
  \[
    \cQ_\pi = (\cU_\pi\circ\cE,\cD\circ\cU_{\pi^{-1}}), 
  \]
  has infidelity $\EE_\pi \widehat{F}\left(\cQ_\pi,\sigma^{\ox\ell}\right)
  \leq \epsilon' \leq \epsilon (\ell+1)^{|J|^2}$
  for the QAVC $\cN$.
\end{proposition}

\begin{proposition}
  \label{prop:compound-to-AVC}
  Let $\cC = \{(\rho_m,D_m):m=1,\ldots,M\}$ be a code of block
  length $\ell$ for the compound channel $\{\cN_\sigma\}_{\sigma\in\cS(J)}$ 
  with error probability $\epsilon$, i.e.~for all $\sigma\in\cS(J)$,
  \[
    P_{\text{err}}(\cC,\sigma^{\ox\ell}) 
         =    \frac1M \sum_{m=1}^M \tr \left( \cN_\sigma^{\ox\ell}(\rho_m)(\1-D_m) \right)
         \leq \epsilon.
  \]
  Then, the random code $(\cC_\pi)_{\pi\in S_\ell}$ with a uniformly
  distributed random permutation $\pi$ of $[\ell]$, defined by
  \[
    \cC_\pi := \{ (U^\pi\rho_m U^{\pi\dagger},U^\pi D_m U^{\pi\dagger}) : m=1,\ldots,M\},
  \]
  has error probability $\epsilon' \leq \epsilon (\ell+1)^{|J|^2}$
  for the QAVC $\cN$.
\end{proposition}

\begin{proof}
We only prove Proposition~\ref{prop:compound-to-AVC}, since 
Proposition~\ref{prop:compound-to-AVC:q} has been argued in
\cite[Appendix~A]{e-ass-Q}, with analogous proofs.
For an arbitrary state $\zeta$ on $J^\ell$, the error probability
of the random code $(\cC_\pi)_{\pi\in S_\ell}$ can be written as
\begin{align}
  \label{eq:S-ell-randomcode}
  \EE_\pi &P_{\text{err}}(\cC_\pi,\zeta) \nonumber\\
%     &= \frac1M \sum_{m=1}^M \EE_\pi 
%                \tr \left( (\cN^{\ox\ell})_\zeta(U^\pi\rho_m U^{\pi\dagger})
%                           (\1-U^\pi D_m U^{\pi\dagger}) \right)              \nonumber\\
     &= \frac1M \sum_{m=1}^M \EE_\pi 
                \tr \left( U^{\pi\dagger}\bigl( \cN^{\ox\ell}(U^\pi\rho_m U^{\pi\dagger},\zeta) \bigr)U^\pi
                                                (\1-D_m) \right)              \nonumber\\
     &= \frac1M \sum_{m=1}^M 
                \tr \left( \cN^{\ox\ell}\bigl( \rho_m, \EE_\pi U^\pi \zeta U^{\pi\dagger} \bigr)
                                                (\1-D_m) \right),
\end{align}
where in the last line we have exploited the $S_\ell$-covariance of the
tensor product channel $\cN^{\ox\ell}$. The crucial feature of the last
expression is that it shows that the error probability that the
jammer can achieve with $\zeta$ is the same as that of the state
\[
  \zeta' = \EE_\pi U^\pi \zeta U^{\pi\dagger}
         = \frac{1}{\ell!} \sum_{\pi\in S_\ell} U^\pi \zeta U^{\pi\dagger}.
\]

This is, by its construction, a permutation-symmetric state, and we can 
apply the de Finetti reduction %(aka ``Postselection Lemma'') 
from \cite{CKR}:
\[
  \zeta' \leq (\ell+1)^{|J|^2} \int_{\sigma\in\cS(J)} \mu({\rm d}\sigma)\, \sigma^{\ox\ell} 
         =:   (\ell+1)^{|J|^2}\mathfrak{F},
\]
with a universal probability measure $\mu$ on the states of $J$, whose
detailed structure is given in \cite{CKR}, but which is not going to be
important for us.
%; $\mathfrak{F}$ could be called the universal de Finetti state.

Indeed, inserting this into the last line of eq.~(\ref{eq:S-ell-randomcode}),
and using complete positivity of $\cN$, we obtain the upper bound
\[\begin{split}
  \EE_\pi &P_{\text{err}}(\cC_\pi,\zeta) 
      =    \frac1M \sum_{m=1}^M \tr \left( \cN^{\ox\ell}\bigl( \rho_m,\zeta' \bigr) (\1-D_m) \right)\\
     &\leq (\ell+1)^{|J|^2} \frac1M \sum_{m=1}^M 
                                \tr \left( \cN^{\ox\ell}\bigl( \rho_m,\mathfrak{F} \bigr) (\1-D_m) \right)\\
     &=    (\ell+1)^{|J|^2} \int_{\sigma\in\cS(J)} \mu({\rm d}\sigma) \frac1M \sum_{m=1}^M 
                             \tr \Bigl( \cN^{\ox\ell}\bigl( \rho_m,\sigma^{\ox\ell} \bigr) \Bigr. \\[-2mm]
     &\phantom{========================}                                 \Bigl. \cdot(\1-D_m) \Bigr) \\
     &=    (\ell+1)^{|J|^2} \int_{\sigma\in\cS(J)} \mu({\rm d}\sigma) P_{\text{err}}(\cC,\sigma^{\ox\ell})
      \leq (\ell+1)^{|J|^2} \epsilon,
\end{split}\]
where in the last step we have used the assumption that for every
jammer state of the form $\sigma^{\ox\ell}$, the error probability is
bounded by $\epsilon$.
\end{proof}

\medskip
To apply this, we need compound channel codes with error 
decaying faster than any polynomial. This is no problem, as there are
several constructions giving even exponentially small error 
for rates arbitrarily close to the compound channel capacity, both
for classical~\cite{BB:compound,Mosonyi:compound}
and quantum codes~\cite{BBN:q-compound}.

\begin{corollary}
  \label{cor:capacities:AVC-to-compound}
  Let $\cN$ be a QAVC. 
  Its classical random coding capacity is given by
  \[
    C_{\text{rand}}(\cN) = C(\{\cN_\sigma\}_\sigma) 
         = \lim_{\ell\rightarrow\infty} 
             \frac{1}{\ell} \max_{\{p_x,\rho_x^{A^\ell}\}} \inf_{\sigma^J} I(X:B^\ell),
  \]
  where $I(X:B^\ell) = S\left(\sum_x p_x\omega_x\right) - \sum_x p_x S(\omega_x)$ 
  is the Holevo information of the ensemble
  $\left\{p_x,\omega_x=\cN^{\ox\ell}(\rho_x\ox\sigma)\right\}$
  \cite{BB:compound,Mosonyi:compound}.

  Similarly, its quantum random coding capacity is
  \[
    Q_{\text{rand}}(\cN) = Q(\{\cN_\sigma\}_\sigma) 
         = \lim_{\ell\rightarrow\infty} 
             \frac{1}{\ell} \max_{\ket{\phi}^{RA^\ell}} \inf_{\sigma^J} I(R\rangle B^\ell),
  \]
  where $I(R\rangle B^\ell) = S(\Omega^{B^\ell})-S(\Omega)$ 
  is the coherent information of the state
  $\Omega = (\id\ox\cN)(\phi^{RA^\ell}\ox\sigma)$
  \cite{BBN:q-compound}.
  \qed
\end{corollary}

%Another consequence is a single-letter formula for the 
%entanglement-assisted capacity.
%\begin{corollary}
%  \label{cor:C_E}
%  For a QAVC, the entanglement-assisted classical capacity is given by
%  the following minimax formula:
%  \[
%    C_E(\cN) = \max_\rho \min_\sigma I(\rho;\cN_\sigma),
%  \]
%  where ...
%  \qed
%\end{corollary}

\section{Capacity dichotomy: \protect\\ Elimination of correlation from random codes}
For classical AVCs or AVQCs with classical jammer, the observations
of Ahlswede \cite{Ahlswede:elim} show that the random coding capacity can
always be attained using at most $O(\log\ell)$ bits of shared randomness. 
This is done by i.i.d.~sampling the shared random variable $\lambda$,
thus approximating, for each channel state $\sigma$, 
$\EE_\lambda P_{\text{err}}(\cC_\lambda,\sigma)$ by an empirical mean 
over $n$ realisations of $\lambda$, except with probability exponentially 
small in $n$. Then, the union bound can be used because the jammer has
``only'' exponential in $\ell$ many choices.
On the face of it, this strategy looks little promising for QAVCs: the
jammer's choices form a continuum, and even if we realise that we can
discretise $\cS(J^\ell)$, any net of states is exponentially large in the
dimension \cite{AliceBob+Banach}, i.e.~doubly exponentially large 
in $\ell$, resulting in a naive bound of $O(\ell)$ for the shared 
randomness required.
However, the linearity of the quantum formalism comes to our rescue.

\begin{observation}
  \label{obs:op}
  From the point of view of the jammer, the error probability of a classical
  code is an observable, $P_{\text{err}}(\cC,\sigma) = \tr\sigma E$, 
  with a POVM element $E=E(\cC)$ depending in a systematic way on the code.
  Likewise, the infidelity of a quantum code can be written
  $\widehat{F}(\cQ,\sigma) = \tr\sigma G$ for a POVM element $G =  G(\cQ)$.
\end{observation}
 
\begin{proof}
  Indeed, using the Heisenberg picture (adjoint map) $\cN^{\ast}$,
  \[\begin{split}
    P_{\text{err}}(\cC,\sigma) 
        &= \frac1M \sum_{m=1}^M \tr \left(\cN^{\ox\ell}(\rho_m\ox\sigma)\right)\!(\1-D_m) \\
%        &= \frac1M \sum_{m=1}^M \tr (\rho_m\ox\sigma)\!\left(\cN^{\ast\ox\ell}(\1-D_m)\right) \\
%        &= \frac1M \sum_{m=1}^M \tr \sigma\!\left[\tr_{A^\ell} (\rho_m\ox\1)\!
%                                                             \left(\cN^{\ast\ox\ell}(\1-D_m)\right)\right] \\
        &\!\!\!\!\!
         = \tr \sigma\!\left[ \frac1M \sum_{m=1}^M \tr_{A^\ell} (\rho_m\ox\1)\!
                                                             \left(\cN^{\ast\ox\ell}(\1-D_m)\right) \right],
  \end{split}\]
  so that 
  \(
    E = \frac1M \sum_{m=1}^M \tr_{A^\ell} (\rho_m\ox\1)\left(\cN^{\ast\ox\ell}(\1-D_m)\right),
  \)
  which is manifestly a POVM element, i.e.~$0 \leq E \leq \1$.

  Likewise, for the infidelity,
  \[\begin{split}
    \widehat{F}(\cQ,\sigma) 
        &= \tr \bigl( (\id\ox\cD\circ\cN^{\ox\ell}\circ\cE)(\Phi_L\ox\sigma) \bigr)\!\cdot\!(\1-\Phi_L) \\
        &\!\!\!\!\!
         = \tr (\Phi_L\ox\sigma)\!\cdot\!
                           \left( (\id\ox\cE^\ast\circ\cN^{\ast\ox\ell}\circ\cD^\ast)(\1-\Phi_L)\right) \\
        &\!\!\!\!\!
         = \tr\sigma G,
  \end{split}\]
  with 
  \(
    G = \tr_{AA'} (\Phi_L\ox\1)\left( (\id\ox\cE^\ast\circ\cN^{\ast\ox\ell}\circ\cD^\ast)(\1-\Phi_L)\right).
  \)

  Obviously, for a random classical code $(\cC_\lambda)$, the expected error
  probability is
  \[
    \EE_\lambda P_{\text{err}}(\cC_\lambda,\sigma) = \tr\sigma (\EE_\lambda E_\lambda),
  \]
  with the POVM elements $E_\lambda = E(\cC_\lambda)$ associated to each
  code $\cC_\lambda$. Likewise for a random quantum code. 
\end{proof}

\medskip
For a random classical code $(\cC_\lambda)$, the jammer's goal is to maximise
the error probability, choosing $\sigma$ in the worst possible way.
But from the present perspective that the error probability is
an observable for Jamie, it is clear that 
$\sup_\sigma \EE_\lambda P_{\text{err}}(\cC_\lambda,\sigma)$ is
simply the maximum eigenvalue of $\overline{E} = \EE_\lambda E_\lambda$.

We say, following general convention, that a random classical or quantum
code $(\cC_\lambda)$ or $(\cQ_\lambda)$ \emph{has error $\epsilon$}
(without reference to any specific state of the jammer) if
\[
  \sup_\sigma \EE_\lambda P_{\text{err}}(\cC_\lambda,\sigma) \leq \epsilon
  \ \text{ or }\ 
  \sup_\sigma \EE_\lambda \widehat{F}(\cQ_\lambda,\sigma) \leq \epsilon,
\]
respectively. By the above discussion is equivalent to
\begin{equation}
  \EE_\lambda E_\lambda \leq \epsilon\1
  \ \text{ or }\ 
  \EE_\lambda G_\lambda \leq \epsilon\1,
\end{equation}
in the sense of the operator order. This is an extremely useful way of
characterising that the random code has a given error.

\bigskip
Our goal now is to select a ``small'' number of $\lambda$'s, say
$\lambda_1,\ldots,\lambda_n$, such that
\begin{equation}
  \label{eq:elimineeschn}
  \frac1n \sum_{\nu=1}^n E_{\lambda_\nu} \leq (\epsilon+\delta)\1,
\end{equation}
ensuring that the random code $(\cC_{\lambda_\nu})_{\nu=1}^n$, with
uniformly distributed $\nu\in[n]$, has error probability $\epsilon+\delta$.
This is precisely the situation for which the matrix tail bounds
in \cite{A-W} were developed. Indeed, quoting \cite[Thm.~19]{A-W},
for i.i.d.~$\lambda_\nu \sim P_\lambda$,
\[
  \Pr\left\{ \frac1n \sum_{\nu=1}^n E_{\lambda_\nu} \not\leq (\epsilon+\delta)\1 \right\}
                        \leq |J|^\ell \cdot \exp\bigl(-n D(\epsilon+\delta\|\epsilon) \bigr),
\]
with the binary relative entropy 
$D(u\|v) = u\ln\frac{u}{v}+(1-u)\ln\frac{1-u}{1-v}$, which can be lower 
bounded by Pinsker's inequality, $D(u\|v) \geq 2(u-v)^2$. Note that 
both the logarithm ($\ln$) and the exponential ($\exp$) are
understood to base $e$. 

Thus, for $n > \frac{1}{2\delta^2}(\ln|J|)\ell$, the right hand
probability bound above is less than $1$, so that there exist 
$\lambda_1,\ldots,\lambda_n$ with (\ref{eq:elimineeschn}).
The number of bits needed to be shared between Alice and Bob
to achieve this, is $\log n$, which we may choose to be
$\leq \log\ell - 2\log\delta + \log\ln|J|$, which is
not zero, but has zero rate as $\ell\rightarrow\infty$.
Exactly the same argument applies to a random quantum code $(\cQ_\lambda)$.
We record this as a quotable statement. 
%Note that it is inherently of one-shot nature: we could as well fix $\ell$, even $\ell=1$.

\begin{proposition}
  \label{prop:q-elimineeschn}
  Let $(\cC_\lambda:\lambda\in\Lambda)$
  be a random classical code of block length $\ell$ for the QAVC 
  $\cN:A\ox J \longrightarrow B$,
  with error probability $\epsilon$.
  Then for $\delta>0$, there exist $\lambda_1,\ldots,\lambda_n \in \Lambda$,
  with $n \leq 1 + \frac{1}{2\delta^2}(\ln|J|)\ell$, such that the
  random code $(\cC_{\lambda_\nu}:\nu\in_R[n])$ has error probability
  $\leq \epsilon+\delta$.
  
  For a random quantum code $(\cQ_\lambda:\lambda\in\Lambda)$,
  with infidelity $\epsilon$, we similarly have that the random code
  $(\cQ_{\lambda_\nu}:\nu\in_R[n])$ has infidelity
  $\leq \epsilon+\delta$.
  \qed
\end{proposition}

\medskip
\begin{remark}
We have discussed here from the beginning the version of the capacity
with average probability of error (and arbitrary encodings). Following
Ahlswede~\cite{Ahlswede:elim} and the generalisation of his method above,
investing another $O(\log \ell)$ bits of shared randomness, or 
loosing $O(\log \ell)$ bits from the code, we can convert any code with 
error $\epsilon$ into one with \emph{maximum error} $\leq 2\epsilon$. 
We omit the details of this argument, as it is exactly as in \cite{Ahlswede:elim}.
\end{remark}

\medskip
Proposition~\ref{prop:q-elimineeschn} allows us to assess the leftmost 
inequalities in the capacity order from eq.~(\ref{eq:cap-chains}).
Because the randomness needed is so little, it can be generated by
a channel code loosing no rate. Hence,
in a certain sense, they are also identities, except in the somewhat
singular case the deterministic classical capacity vanishes:

\begin{corollary}
  \label{cor:dichotomy}
  The classical capacity of a QAVC $\cN$ is either $0$ or, if
  it is positive, it equals the random coding capacity:
  \[
    C_{\text{det}}(\cN) = \begin{cases} C_{\text{rand}}(\cN) & \text{if } C_{\text{det}}(\cN)>0, \\
                                                           0 & \text{otherwise}. 
                          \end{cases}
  \]
  Similarly, for the quantum capacity:
  \[
    \phantom{=====}
    Q_{\text{det}}(\cN) = \begin{cases} Q_{\text{rand}}(\cN) & \text{if } C_{\text{det}}(\cN)>0, \\
                                                           0 & \text{otherwise}. 
                          \phantom{=======:}\qed
                          \end{cases}
  \]  
\end{corollary}

\section{Discussion and Outlook}
We have shown that in a fully quantum jammer channel model (QAVC),
the random coding capacity, for both quantum and classical transmission,
can be reduced to the capacity of a corresponding compound channel;
furthermore, by extension of the ``elimination of correlation''
technique, that the shared randomness required has zero rate, thus
implying dichotomy theorems for the deterministic classical and
quantum capacities.
Since the derandamisation leaves so little randomness,
we can apply the results also to say something about the identification
capacity of QAVCs: Either the ID-capacity vanishes, or it 
equals the random coding capacity $C_{\text{rand}}(\cN)$.

Our work leaves two important open questions: First, to give 
necessary and sufficient conditions for vanishing classical capacity.
For classical AVCs this is the co-called ``symmetrizability''
\cite{Ericsson,CN:symmetrize}.
But what is the analogue of this condition for quantum channels?

Second, both parts of our reasoning relied on the finite dimensionality
of the jammer system $J$.
It is not so clear how to deal with infinite dimension of $J$,
on the other hand. A priori we have a problem already in
Proposition~\ref{prop:compound-to-AVC}, since the de Finetti reduction 
has an upper bound depending on the dimension $|J|$.
However, one can prove the random coding capacity theorem directly
from first principles, without recourse to de Finetti reductions.
%[CITE?].

Then, we have the problem again in the derandomisation step, which
requires bounded $|J|$ to apply the matrix tail bound. 
We need some kind of quantum net argument
to be able to go to a finite dimensional subspace $J' < J$ that 
somehow approximates the relevant features of $\cN$ up to 
error $\eta$ and block length $\ell$. 
Classically, the finiteness of the alphabet of channel states is
irrelevant, as long as we have finite sender and receiver 
alphabets. The reason is that for each block length $\ell$ we can 
choose a subset of channel states of size polynomial in $\ell$,
corresponding to an $\frac{\eta}{\ell}$-net of channels realised by the 
jammer, for any fixed $\eta > 0$. Indeed, for the QAVC with classical jammer,
which may be described by a state set $\cS$, the following statements
are easily obtained by standard methods.

\begin{lemma}
  \label{lemma:discretise-S}
  For every $\eta>0$, there exists a set $\cS' \subset \cS$ of
  cardinality $|\cS'| \leq \left( \frac{10|A|^2}{\eta} \right)^{2|A|^2|B|^2}$, 
  with the property that for every $s\in\cS$ there is an $s'\in\cS'$ with
  $\frac12 \|\cN_s-\cN_{s'}\|_\diamond \leq \eta$, where the
  norm is the diamond norm (aka completely bounded trace norm)
  on channels~\cite{diamond,Watrous:diamond}.
  \qed
\end{lemma}

By applying this lemma with $\frac{\eta}{\ell}$, the ``telescoping trick''
and the triangle inequality to bound
$\frac12 \left\| \cN^{\ox\ell}_{s^\ell} - \cN^{\ox\ell}_{\sigma^\ell} \right\|_\diamond$
for $s^\ell\in\cS^\ell$ and $\sigma^\ell\in{\cS'}^\ell$, we obtain then:

\begin{lemma}
  \label{lemma:discretise-jammer}
  For every $\eta>0$ and integer $\ell$, there exists a subset
  $\cS' \subset \cS$ of cardinality
  $|\cS'| \leq \left( \frac{10|A|^2\ell}{\eta} \right)^{2|A|^2|B|^2}$, such that
  \[\begin{split}
    \sup_{\sigma^\ell\in{\cS'}^\ell} \EE_\lambda P_{\text{err}}(\cC,\sigma^\ell) 
        &\leq \sup_{s^\ell\in{\cS}^\ell} \EE_\lambda P_{\text{err}}(\cC,s^\ell)          \\
        &\leq \sup_{\sigma^\ell\in{\cS'}^\ell} \EE_\lambda P_{\text{err}}(\cC,\sigma^\ell)
  \end{split}\]
  for any random code $(\cC_\lambda:\lambda\in\Lambda)$. 
  Similar for the infidelity of random quantum codes.
  \qed
\end{lemma}

%The bound on the size of $\cS'$ makes it that the impact on the 
%derandomisation is minimal, since the alphabet size only enters 
%doubly logarithmically, i.e.~adding at most another term $O(\log\ell)$.

Since we need to entangle both the $A$'s and the $J$'s, 
it seems that the most natural approach is to answer the following question.

\begin{question}
\label{question-four}
\normalfont
Let $\cN:\cL(A\ox J)\longrightarrow \cL(B)$ be a cptp map with
finite dimensional $A$ and $B$, and $\eta>0$. Is it possible
to find a subspace $J' \subset J$ of dimension bounded
by some polynomial in $\eta^{-1}$, with the following property?

\begin{quote}
For every Hilbert space $K$ and state $\sigma$ on $J\ox K$,
there exists another state $\sigma'$ on $J'\ox K$ such that
$\frac12 \| \cN_\sigma - \cN_{\sigma'} \|_\diamond \leq \eta$.
\end{quote}
Here, $\cN_\sigma$ and $\cN_{\sigma'}$ are channels from $A$ to $B\ox K$, 
defined by inserting the respective state into the jammer register:
\begin{align*}
  \cN_\sigma(\rho)    := (\cN\ox\id_K)(\rho\ox\sigma), \ 
  \cN_{\sigma'}(\rho) := (\cN\ox\id_K)(\rho\ox\sigma').
\end{align*}
\end{question}

We can reduce this to the more elementary question of approximating 
the output of the ``Choi channel'' $\Gamma:\cL(J) \rightarrow \cL(C)$, 
with $C=A\ox B$,
defined by $\Gamma(\sigma) = (\id_A\ox\cN)(\Phi^{AA'}\ox\sigma)$,
mapping each $\sigma$ to the Choi state of the channel $\cN_\sigma$:
Namely, the question is whether for every Hilbert space $K$ and 
state $\sigma$ on $J\ox K$, does there exist a state 
$\sigma'$ on $J'\ox K$ such that
\(
  \displaystyle
  \frac12 \| (\Gamma\ox\id)(\sigma-\sigma') \|_1 \leq \widetilde{\eta} := \eta/|A|^2 ?
\)

\medskip
We now show that a positive answer to Question~\ref{question-four}, 
with deviation $\frac{\eta}{\ell}$, could be used to replace
the $\ell$ environments of $\cN^{\ox\ell}$ in $\ell$ steps each by a
finite dimensional approximation. In this way, we would be able to
find, for every state $\sigma$ on $J^\ell$, another state $\sigma'$ on 
${J'}^\ell$, with 
\begin{equation}
  \label{eq:approxi}
  \frac12 \left\| \left(\cN^{\ox\ell}\right)_\sigma 
                  - \left(\cN^{\ox\ell}\right)_{\sigma'} \right\|_\diamond \leq \eta.
\end{equation}

\begin{proof}
Set $\sigma^{(0)} := \sigma$; we shall define a sequence of
approximants $\sigma^{(i)}$ on ${J'}^{\ox i}\ox J^{\ell-i}$ 
($i=1,\ldots,\ell$), as follows:

To obtain $\sigma^{(1)}$, we apply Question~\ref{question-four} with
$K=J^{\ox\ell-1}$ (the last $\ell-1$ of the $J$-systems) to obtain
\[
  \frac12 \left\| \cN^{[1]}_{\sigma^{(0)}} - \cN^{[1]}_{\sigma^{(1)}} \right\|_\diamond 
        \leq \frac{\eta}{\ell},
\]
where the notation $\cN^{[i]} = \id^{\ox i-1} \ox \cN \ox \id^{\ox \ell-i}$
indicates application of the channel to the $i$-th system in $J^\ell$.
Proceeding inductively, assume that we already have constructed a state
$\sigma^{(i-1)}$ on ${J'}^{\ox i-1} \ox J^{\ox \ell-i+1}$, Question~\ref{question-four}
applied to $K={J'}^{\ox i-1} \ox J^{\ox \ell-i}$ (i.e.~all the $J'$ systems and
the last $\ell-i$ of the $J$'s) gives us a state
$\sigma^{(i+1)}$ on ${J'}^{\ox i} \ox J^{\ox \ell-i}$ such that
\[
  \frac12 \left\| \cN^{[i]}_{\sigma^{(i-1)}} - \cN^{[i]}_{\sigma^{(i)}} \right\|_\diamond 
         \leq \frac{\eta}{\ell}.
\]

Since the diamond norm is contractive under composition with cptp maps,
we obtain for all $i=1,\ldots,\ell$ that
\[
  \frac12 \left\| \left(\cN^{\ox\ell}\right)_{\sigma^{(i-1)}} 
                - \left(\cN^{\ox\ell}\right)_{\sigma^{(i)}} \right\|_\diamond \leq \frac{\eta}{\ell},
\]
and via the triangle inequality we arrive at eq.~(\ref{eq:approxi}), by
letting $\sigma' := \sigma^{(\ell)}$ and recalling $\sigma = \sigma^{(0)}$.
\end{proof}

\medskip
This would mean that any behaviour that the jammer
can effect by choosing states on $J^\ell$, can be approximated
up to $\pm\eta$ (on block length $\ell$) by choices from ${J'}^\ell$,
analogously to Lemma~\ref{lemma:discretise-jammer}, which actually
provides a positive answer to Question \ref{question-four} in the case
of a classical jammer. 
Since $|J'|$ is bounded polynomially in $\ell$, we could apply 
now Proposition \ref{prop:q-elimineeschn} and incur an additional 
term of $O(\log\ell)$ in the shared randomness required, in particular
it will still be of zero rate.

\medskip
A third complex of questions concerns the extension of the present
results to other quantum channel capacities. This is easy along
the above lines for cases like the entanglement-assisted capacity 
(cf.~\cite{BJK:E-assisted,DSW:feedback-0}), but challenging for
others, such as the private capacity~\cite{CWY,Devetak}. 
This is interesting because
the error criterion (of decodability and privacy) does not seem
to correspond to an observable on the jammer system. We leave this
and the other open problems for future investigation.

\bigskip
\textbf{Acknowledgments.}
HB and CD were supported by the German
BMBF through grants~16KIS0118K and 16KIS0117K.
JN was supported by the German BMWi and ESF, grant~03EFHSN102.
AW was supported by the ERC Advanced Grant IRQUAT,
the Spanish MINECO, projects no. FIS2013-

\vfill\pagebreak\noindent
40627-P and FIS2016-86681-P, 
%with the support of FEDER funds, 
and the Generalitat de Catalunya, CIRIT project 2014-SGR-966.

\end{document}